\documentclass[preprint, 11pt, final]{elsarticle}

\usepackage{fancyvrb}
\usepackage[T1]{fontenc} 
\usepackage{microtype} 
\usepackage{amssymb, amsmath, amsthm} 
\numberwithin{equation}{section} 
\usepackage{graphicx}

\usepackage[hmarginratio=1:1,top=32mm,columnsep=20pt]{geometry} 
\usepackage{multicol} 
\usepackage[hang, small,labelfont=bf,up,textfont=it,up]{caption} 
\usepackage{booktabs} 
\usepackage{float} 
\usepackage[colorlinks=true, urlcolor=green, pdfborder={0 0 0}]{hyperref} 

\usepackage{paralist} 

\usepackage[linesnumbered,algoruled]{algorithm2e} 
\usepackage{enumitem} 
\usepackage{subfig} 

\usepackage{setspace, ifdraft, blindtext} 

\ifdraft{\usepackage{endfloat}}{}

\usepackage{changes}

\usepackage{lineno}

\makeatletter
\let\start@align@nopar\start@align
\let\start@gather@nopar\start@gather
\let\start@multline@nopar\start@multline
\long\def\start@align{\par\start@align@nopar}
\long\def\start@gather{\par\start@gather@nopar}
\long\def\start@multline{\par\start@multline@nopar}
\makeatother

\usepackage[numbers]{natbib}

\setcitestyle{authoryear, open={(}, close={)} }

\def\ie{\rm{{\it i.e.\ }}}
\def\eg{\rm{{\it e.g.\ }}}

\newcommand{\M}{\mathcal{M}}
\newcommand{\D}{\mathcal{D}}
\renewcommand{\L}{\mathcal{L}}
\newcommand{\R}{\mathbb{R}}
\newcommand{\bs}{\boldsymbol}
\newcommand{\btheta}{\boldsymbol\theta}

\newtheorem{theorem}{Theorem}

\theoremstyle{definition}

\theoremstyle{remark}

\graphicspath{{img/}}
\journal{Unknown}

\begin{document}

\begin{frontmatter}

\title{\small{\color{red} Please cite as: F.A. DiazDelaO, A. Garbuno-Inigo, S.K.
Au, I. Yoshida, Bayesian updating and model class selection with Subset
Simulation, {\em Computer Methods in Applied Mechanics and Engineering}, Volume
317, 2017}\\Bayesian updating and model class selection with Subset Simulation}

\author[liv]{F.~A.~DiazDelaO \corref{cor1}}
\author[liv]{A.~Garbuno-Inigo}
\author[liv]{S.~K.~Au}
\author[tok]{I.~Yoshida}

\cortext[cor1]{Corresponding author}

\address[liv]{Institute for Risk and Uncertainty, School of Engineering,
 University of Liverpool\\ Brownlow Hill, Liverpool L69 3GH, United Kingdom.}

\address[tok]{Department of Urban and Civil Engineering, Tokyo City University.
 \\1-28-1 Tamazutsumi Setagaya-ku, Tokyo 158-8557, Japan.}

\begin{abstract}

Identifying the parameters of a model and rating competitive models based on
measured data has been among the most important and challenging topics in modern
science and engineering, with great potential of application in structural
system identification, updating and development of high fidelity models. These
problems in principle can be tackled using a Bayesian probabilistic approach,
where the parameters to be identified are treated as uncertain and their
inference information are given in terms of their posterior probability
distribution. For complex models encountered in applications, efficient
computational tools robust to the number of uncertain parameters in the problem
are required for computing the posterior statistics, which can generally
be formulated as a multi-dimensional integral over the space of the uncertain
parameters. Subset Simulation has been developed for solving reliability
problems involving complex systems and it is found to be robust to the number of
uncertain parameters. An analogy has been recently established between a
Bayesian updating problem and a reliability problem, which opens up the
possibility of efficient solution by Subset Simulation. The formulation, called
BUS (Bayesian Updating with Structural reliability methods), is based the
standard rejection principle. Its theoretical correctness and efficiency
requires the prudent choice of a multiplier, which has remained an open
question. This paper presents a fundamental study of the multiplier and
investigates its bias effect when it is not properly chosen. A revised
formulation of BUS is proposed, which fundamentally resolves the problem such that Subset Simulation can be 
implemented without knowing the multiplier a priori. An automatic stopping condition is also provided. 
Examples are presented to illustrate the theory and applications.

\end{abstract}

\begin{keyword}
Bayesian inference \sep BUS \sep Subset Simulation \sep Markov Chain Monte
Carlo \sep model updating
\end{keyword}

\end{frontmatter}

\section{Introduction} 

Making inference about the parameters of a mathematical model based on observed
measurements of the real system is one of the most important problems in modern
science and engineering. The Bayesian approach provides a fundamental means to
do this in the context of probability logic
\citep{Malakoff1999,Richard1961,Jaynes2003}, where the parameters are viewed as
uncertain variables and the inference results are cast in terms of their
probability distribution after incorporating information from the observed data.
In engineering dynamics, for example, vibration data from a structure is
collected from sensors and used for identifying the modal properties (\eg
natural frequencies, damping ratios, mode shapes) and structural model
properties (\eg stiffness, mass) \citep{Hudson1977,Ewins2000}. This has been
formulated in a Bayesian context \citep{Beck1998,Beck2010}, which
resolved a number of philosophically challenging issues of the inverse problem,
such has the treatment of multiple sets of parameters giving the same model fit
to the data, an issue known as {\it identifiability}.

Let $\Theta \in \mathbb{R}^n$ be a set of parameters of a model $\M$,
based on which a probabilistic prediction of the data $\D$ can be
formulated through the likelihood function $P(\D| \bs
\theta, \M)$. Clearly, the probability distribution of $\Theta$ depends
on the available information. Based only on knowledge in the context of
$\M$, the distribution is described by the prior distribution
$P(\bs\theta | \M)$. When data about the system is available,
it can be used to update the distribution. Using Bayes' Theorem, the posterior distribution that 
incorporates the data information in the context of
$\M$ is given by
\begin{align}
P(\bs\theta | \D, \M ) = P(\D|\M)^{-1} \, P(\D| \bs \theta, \M) \, P(\bs\theta | \M), \label{eq:posterior}
\end{align}

\noindent where 
\begin{align}
 P(\D|\M) = \int_{\bs\Theta} P(\D| \bs \theta, \M) \, P(\bs\theta | \M) \, d\bs\theta, \label{eq:normalising}
\end{align}

\noindent is a normalizing constant. Future predictions of a response quantity
of interest, say $r(\bs\theta)$, can be updated by incorporating data
information, through the posterior expectation \citep{Papadimitriou2001}:
\begin{align}
E[r(\bs\theta | \D, \M)] = \int r(\bs\theta) \, P(\bs\theta | \D, \M) \, d\bs\theta. 
\label{eq:expectation_model_data}
\end{align}

As far as the posterior distribution of $\bs\theta$ for a given model $\M$ is
concerned, the constant in Eq. \eqref{eq:normalising} is immaterial because it
does not change the distribution. However, It is the primary quantity of study
in Bayesian model class selection problems where competing models are compared
based on the value of $P(\M) P(\D|\M)$ \citep{Carlin1995,Chen2012,Beck2004a}. In
that context, $P(\D|\M)$ is often called the {\it evidence} (the higher the
better).

Capturing efficiently essential information about the posterior distribution,
\ie posterior statistics, and calculating the posterior expectation is a non-
trivial problem, primarily resulting from the complexity of the likelihood
function. In many applications, the likelihood function is only implicitly
known, \ie its value can be calculated point-wise but its dependence on the
model parameters is mathematically intractable. This renders analytical
solutions infeasible and conventional numerical techniques inapplicable. In this
case, Markov Chain Monte Carlo (MCMC)
\citep{Metropolis1953,Hastings1970,Robert2004,Fishman1996a} is found to provide
a powerful computational tool. MCMC allows the samples of an arbitrarily given
distribution to be efficiently generated as the samples of a specially designed
Markov chain. In MCMC, candidate samples are generated by a {\it proposal
distribution} (chosen by the analyst) and they are adaptively accepted based on
ratios of the target distribution value at the candidate and the current sample.

While MCMC in principle provides a powerful solution for Bayesian computation,
difficulties are encountered in applications, motivating different variants of
the algorithm. For example, in problems with a large amount of data, the
posterior distribution takes on significant values only in a small region of the
parameter space, whose size generally shrinks in an inverse square root law with
the data size. Depending on sufficiency or relevance of the data for the model
parameters, the regions of significant probability content can be around a set
of isolated points (globally or locally identifiable) or a lower dimensional
manifold (unidentifiable) with non-trivial geometry
\citep{Katafygiotis1998,Katafygiotis2002a}. To the least extent this causes
efficiency problems, making the choice of the proposal distribution difficult
and leading to high rejection rate of candidates and hence poor efficiency. When
the issue is not managed, significant bias can result in the statistical
estimation based on the samples. Strategies similar to simulated annealing have
been proposed to convert the original difficult updating problem effectively
into a sequence of more manageable problems with less data, thereby allowing the
samples to adapt gradually \citep{Beck2002,Cheung2010,Ching2007}. Another issue
is {\it dimension sustainability}, \ie whether the algorithm remains applicable
when the number of variables (\ie dimension) of the problem increases. This
imposes restrictions on the design of MCMC algorithms so that quantities such as
the ratio of likelihood functions involved in the simulation process do not {\it
degenerate} as the dimension of the problem increases.

Application robustness and dimension sustainability are well-recognized in the
engineering reliability method literature
\citep{Au2003,Schueeller2004,Katafygiotis2008}. In this area, the general
objective is to determine the failure probability that a scalar response of
interest exceeds a specified threshold value, or equivalently to determine its
complementary cumulative distribution function (CCDF) near the upper tail (\ie
large thresholds). Subset Simulation (SuS) \citep{Au2001,Au2014} has been
developed as an advanced Monte Carlo strategy that is efficient for small
failure probabilities (rare events) but still retain a reasonable robustness
similar to the Direct Monte Carlo method. In SuS, samples conditional on a
sequence of intermediate failure events are generated by MCMC and they gradually
populate towards the target failure region. These {\it conditional samples}
provide information for estimating the whole CCDF of the response quantity of
interest. SuS typically does not make use of any problem-specific information,
treating the input-output relationship between the response and the uncertain
parameters as a {\it black box}. Based on an independent-component MCMC
strategy, it is applicable for an arbitrary (potentially infinite) number of
uncertain variables in the problem.

By establishing an analogy with the reliability problem that SuS is originally
designed to solve, it is possible to adapt SuS to provide an efficient solution
for another class of problems. For example, by considering an {\it augmented
reliability problem} where deterministic design parameters are artificially
considered as uncertain, SuS has been applied to investigate the sensitivity of
the failure probability with respect to the design parameters and their optimal
choice without repeated simulation runs
\citep{Au2005,Ching2007b,Song2009,Taflanidis2009a}. Another example can be found
in constrained optimization problems, where an analogy was established between
rare failure events in reliability problems and extreme events in optimization
problems, allowing SuS to be applied to solving complex problems with nonlinear
objective functions and potentially a large number of inequality constraints and
optimization variables \citep{Li2010,Qi2011}.

In view of the application robustness and dimension sustainability, it would be
attractive to adapt SuS for Bayesian computations. This is not trivial since the
problem contexts are different. One major difference is that in the reliability
problem the uncertain parameters follow standard classes of distributions (\eg
Gaussian, exponential) specified by the analyst; while in the Bayesian updating
problem the uncertain parameters follow the posterior distribution, which
generally does not belong to any standard distribution because the likelihood
function is problem-dependent.

Recent developments have shown promise for adapting SuS to Bayesian updating
problems. In the context of Approximate Bayesian Computation (ABC),
\cite{Chiachio2014} built an analogy with the reliability problem so that the
posterior samples in the Bayesian updating problem can be obtained as the
conditional samples in SuS at the highest simulation level determined by a
tolerance parameter that gradually diminishes. The latter controls the
approximation of the likelihood function through a proximity model (a feature of
ABC) between the measured and simulated data for a given value of model
parameter.

Along another line of thought, \cite{Straub2014} recently provided a formulation
called BUS (Bayesian Updating using Structural reliability methods) that opens
up the possibility of Bayesian updating using SuS. It combined an
earlier idea \citep{Straub2011} with the standard rejection principle to
establish an analogy between a Bayesian updating problem and a reliability
problem, or more correctly a {\it probabilistic failure analysis} problem
\citep{Au2003,Au2004,Au2014}. Through the analogy, the samples following the
posterior distribution in the Bayesian updating problem can be obtained as the
conditional samples in the reliability problem. Unlike ABC, the formulation is
exact as it respects fully the original likelihood function; and in this sense
it is more fundamental. One outstanding problem, however, is the choice of the
{\it likelihood multiplier}, or {\it multiplier} in short, in the context of
rejection principle. To guarantee the theoretical correctness of the analogy, it
must be less than the reciprocal of the maximum value of the likelihood
function, which is generally unknown especially before the problem is solved.
Some suggestions have been given in \cite{Straub2014} based on inspection of the
likelihood function. An adaptive choice was suggested based empirically on the
generated samples \citep{Betz2014}. It is more robust to applications as it
does not require prior input from the analyst. It offers no guarantee on
correctness, however, due to the incomplete nature of finite sampling
information which seems inevitable. The problem with the choice of the
multiplier remains open.

This work is motivated by the choice of the multiplier and more fundamentally
its mathematical and philosophical role in the BUS formulation. A rigorous
mathematical study is carried out to provide fundamental understanding of the
multiplier, which leads to a revised BUS formulation allowing SuS to be
implemented independent of the choice of the multiplier and convergence of
results to be checked formally. Essentially, by defining the failure event in
the BUS formulation, we show that SuS can in fact be implemented {\it without
the multiplier} and the samples beyond a certain simulation level all have the
same target posterior distribution.

This paper is organized as follows. We first give an overview of SuS and the original BUS formulation. 
The mathematical role of the multiplier and its bias effect arising from inappropriate choice are then
investigated. A revised formulation is then proposed and associated theoretical
issues are investigated, followed by a discussion on the application of SuS under
the revised formulation. Examples are presented to explain the theory and
illustrate its applications.

\section{Subset Simulation} \label{sec:sus}

We first briefly introduce Subset Simulation (SuS) to facilitate understanding
its application in the context of Bayesian model updating and model class
selection later. SuS is an advanced Monte Carlo method for reliability and
failure analysis of complex systems, especially for rare events. It is based on
the idea that a small failure probability can be expressed as a product of
larger conditional failure probabilities, effectively converting a rare
simulation problem into a series of more frequent ones.

\subsection{Reliability and failure analysis problem}

Despite the variety of failure events in applications, they can often be
formulated as the exceedance of a critical response over a specified threshold.
Let $Y = h(\bs\theta)$, be a scalar response quantity of interest that depends
on the set of uncertain parameters $\bs\theta$ distributed as the parameter
probability density function (PDF) $q(\bs\theta)$. The function $h(\cdot)$
represents the relationship between the uncertain input parameters and the
output response. The parameter PDF $q(\cdot)$ is specified by the analyst from
standard distributions. Without loss of generality, the uncertain parameters are
assumed to be continuous-valued and independent, since discrete-valued
variables or dependent variables can be obtained by mapping continuous-valued
independent ones.

The primary interest of reliability analysis is to determine the {\it failure
probability} $P(Y>b)$ for a specified threshold value $b$:
\begin{align}
P(Y>b) = \int q(\bs\theta) \, I(\bs\theta \in F) \, d\bs\theta, 
\end{align}

\noindent where
\begin{align}
F = \{Y > b\} = \{\bs\theta \in \R^n : h(\bs\theta) > b \}, 
\end{align}

\noindent denotes the failure event or the failure region in the parameter space,
depending on the context; $I(\cdot)$ is the indicator function, equal to 1 if
its argument is true and zero otherwise. Probabilistic failure analysis on the
other hand is concerned with what happens when failure occurs, which often
involves investigating the expectation of some response quantity $r(\bs\theta)$
(say) conditional on the failure event, \ie
\begin{align}
E[r(\bs\theta)|F] = \int r(\bs\theta) \, q(\bs\theta|F) \, d\bs\theta, 
\end{align}

\noindent where
\begin{align}
q(\bs\theta| F) = P_F^{-1} \, q(\bs\theta) \, I(\bs\theta \in F), 
\end{align}

\noindent is the PDF of  conditional on failure. 

When the relationship between $Y$ and $\bs\theta$, \ie the function $h(\cdot)$,
is complicated, analytical or conventional numerical integration is not feasible
for computing $P(Y>b)$ or $E[r(\bs\theta)|F]$ and thus advanced computational
methods are required for their efficient determination. SuS offers an efficient
solution by generating a sequence of sample populations of $\bs\theta$
conditional on increasingly rare failure events $\{Y > b_i\}$, where $\{b_i :
1,2,\ldots\}$ is an increasing sequence of threshold values adaptively
determined during the simulation run. These {\it conditional samples} provide
information for estimating the CCDF of $Y$, \ie  $P(Y>b)$ versus $b$ from the
frequent (left tail) to the rare (right tail) regime. When the right tail covers
the threshold value associated with the target failure event, the required
failure probability can be obtained from the estimate of the CCDF. The
conditional samples can also be used for estimating the conditional expectation
in probabilistic failure analysis, a feature not shared by conventional variance
reduction techniques. As we shall see in the next section, the conditional
samples provide the posterior samples required for Bayesian model updating. The
failure probability provides the information for estimating the evidence for
Bayesian model class selection.

\subsection{Subset Simulation procedure}

A typical SuS algorithm is presented as follows \citep{Au2001,Au2014}. Two
parameters should be set before starting a simulation run: 1) the {\it level
probability} $p_0 \in (0,1)$ and 2) the {\it number of samples per level} $N$.
It is assumed that $p_0N$ and $p_0^{-1}$ are positive integers. As will be seen
shortly, these are respectively equal to the number of chains and the number of
samples per chain at a given simulation level. In the reliability literature, a
prudent choice is $p_0 = 0.1$. The number of samples $N$ controls the
statistical accuracy of results (the higher the better), generally in an inverse
square root manner. Common choice ranges from a few hundreds to over a thousand,
depending on the target failure probability.

A simulation run starts with Level 0 (unconditional), where $N$ i.i.d.
(independent and identically distributed) samples of $\bs\theta$ are generated
from $q(\cdot)$, \ie direct Monte Carlo. The corresponding values of $Y$ are
computed and arranged in ascending order, giving an ordered list denoted by
$\{b_k^{(0)}: {\color{black}k = 1, \ldots, N\}}$. The value $b_k^{(0)}$ gives the estimate of $b$
corresponding to the exceedance probability $p_k^{(0)} = P(Y > b)$ where
\begin{align}
p_k^{(0)} = \frac{N-k}{N}, & & k = 1, \ldots, N.
\end{align}

\noindent The next level, \ie Level 1, is conditional on the intermediate
failure event $\{Y > b_1\}$, where $b_1$ is determined as the $(p_0N+1)$-th
largest sample value of $Y$ at Level 0, \ie
\begin{align}
b_1 = b_{N(1-p_0)}^{(0)}. 
\end{align}
	
\noindent By construction, the $p_0N$ samples of $\bs\theta$ corresponding to
$\{b_{N(1-p_0)+j}^{(0)}: j = 1, \ldots, p_0N\}$ are conditional on $\{Y>b_1\}$.
These conditional samples are used as {\it seeds} for generating additional
samples conditional on $\{Y>b_1\}$ by means of MCMC. A MCMC chain of $p_0^{-1}$
samples is generated from each seed, giving a total population of $p_0N \times
p_0^{-1} = N$ samples conditional on $\{Y>b_1\}$ at Level 1.

During MCMC the values of $Y$ of the conditional samples at Level 1 have been
calculated. They are arranged in ascending order, giving the ordered list
denoted by $\{b_1^{(1)}: k = 1, \ldots, N\}$. The value $b_k^{(1)}$ gives the
estimate of $b$ corresponding to exceedance probability $p_k^{(1)} = P(Y> b)$
where
\begin{align}
p_k^{(1)} = p_0 \, \frac{N-k}{N}, & & k = 1, \ldots, N.
\end{align}

\noindent The next level, \ie Level 2, is conditional on $\{Y > b_2\}$ where
$b_2$ is determined as the $(p_0N+1)$-th largest sample value of $Y$ at Level 1,
\ie
\begin{align}
b_2 = b_{N(1-p_0)}^{(1)}. 
\end{align}

The above process of generating additional MCMC samples and moving up simulation
levels is repeated until the target threshold level or probability level has
been reached. In general, at Level $i$ ($i=1, \ldots, N$), in the ordered list
of sample values of $Y$ denoted by $\{b_k^{(i)}: k = 1, \ldots, N\}$, the value
$b_k^{(i)}$ gives the estimate of $b$ corresponding to exceedance probability
$p_k^{(i)} = P(Y> b)$ where
\begin{align}
p_k^{(i)} = p_0^i \, \frac{N-k}{N}, & & k = 1, \ldots, N.
\end{align}

Several features of SuS are worth-mentioning. It is population-based in the
sense that the samples at a given level are generated from multiple $(p_0N)$
chains, making it robust to ergodic problems. An independent-component MCMC
algorithm is used, which is the key to be sustainable for high dimensional
problems \citep{Au2001,Schueeller2004,Haario2005}. The conditional samples at
each level all have the target conditional distribution and there is no {\it
burn-in} problem commonly discussed in the MCMC literature. This is because the
MCMC chains are all started with a seed distributed as the target distribution
(conditional on that level), and so they are stationary right from the start.

Variants of the SuS algorithm have been proposed to improve efficiency, \ie
\cite{Papadopoulos2012,Zuev2011,Bourinet2011}. See also the review in Section
5.9 of \cite{Au2014}. The algorithm can even be implemented as a VBA (Visual
Basic for Applications) Add-In in a spreadsheet \citep{Au2010,Wang2010a}.

\section{BUS formulation}

In this section we briefly review the BUS formulation \citep{Straub2014,Straub2016} that
builds an analogy between the Bayesian updating problem and a reliability
problem, thereby allowing SuS to be applied to the former. For mathematical
clarity and to simplify notation, in the Bayesian updating problem we use
$q(\btheta)$ to denote the prior PDF $p(\btheta)$, $\L(\btheta)$ to denote the
likelihood function $p(\btheta | \D,\M)$, $P_{\D}$ to denote the normalizing
constant $P(\D|\M)$, and $p_{\D}(\btheta)$ to denote the posterior PDF. The
same symbol $q(\btheta)$ is used for the prior PDF in the Bayesian updating
problem and the parameter PDF in the reliability problem, as it has the same
mathematical property (chosen from standard distributions by the analyst) and
role (the distribution to start the SuS run) in both problems. In a Monte Carlo
approach the primary target in Bayesian model updating is to generate samples
according to the posterior PDF $p_{\D}(\btheta)$ (rewritten from
\eqref{eq:posterior}):
\begin{align}
p_{\D}(\btheta) = P_{\D}^{-1} \, q(\btheta) \, \L(\btheta). \label{eq:rej_posterior}
\end{align}

\subsection{Rejection Principle}

The BUS formulation is based on the conventional rejection principle. Let $c$,
called the {\it likelihood multiplier} in this work, or simply {\it multiplier},
be a scalar constant such that for all $\btheta$ the following inequality holds:
\begin{align}
c \, \L(\btheta) \leq 1. \label{eq:reject_c}
\end{align}

\noindent Also, assume that i.i.d. samples can be efficiently generated from the
prior PDF $q(\btheta)$. This is a reasonable assumption because the prior PDF is
often chosen from a standard class of distributions (\eg Gaussian,
exponential). In the above context, a sample $\btheta$ distributed as the
posterior PDF $p_{\D}(\btheta) \propto q(\btheta) \L(\btheta)$ in
\eqref{eq:rej_posterior} can be generated from the following straightforward
application of the rejection principle:

\medskip
Step 1. Generate $U$ uniformly distributed on $[0,1]$ and $\btheta$
distributed with the prior PDF $q(\btheta)$.

Step 2. {\color{black}If $U < c \L(\btheta)$}, return $\btheta$ as the sample. Otherwise
go back to Step 1.

\noindent It can be shown \citep{Straub2014} that the sample $\btheta$ returned
from the above algorithm is distributed as $p_{\D}(\btheta)$, that is by
marginalising the auxiliary component $u$ as
\begin{align}
p_{\btheta'}(\btheta) & = \int_0^1 p_{\btheta', u} (\btheta,u) \, du \propto p_{\D}(\btheta). 
\label{eq:marginal_failure}
\end{align}

Although the above rejection algorithm is theoretically viable, the acceptance
probability and hence efficiency is often very low in typical updating problems
with a reasonable amount of data. This is because a sample drawn from the prior
PDF $q(\btheta)$ often has a low likelihood value $\L(\btheta)$ when the data is
informative about the uncertain parameters, leading to significant change from
the prior to the posterior PDF.

\subsection{Equivalent reliability problem}

Recognizing the high rejection rate when the rejection principle is directly
applied, BUS transforms the problem into a reliability problem. The premise is
that this will allow the existing algorithms developed in the reliability method
literature to be applied to Bayesian updating problems, especially those are
that capable of generating samples from the frequent (safe) region to the rare
(failure) region, such as SuS. The reliability problem analogy of the Bayesian
updating problem is constructed as follows. Consider a reliability problem with
uncertain parameters $(\btheta, u)$ having the joint PDF $q(\btheta) \, I(0\leq
u\leq1)$, where the {\it failure event} is defined as
\begin{align}
F = \{U < c \, \L(\btheta)\}. \label{eq:equiv_rel}
\end{align}

\noindent Suppose that by some means (\eg SuS) we can obtain a {\it failure
sample} distributed as $q(\btheta) \, I(0\leq
u\leq1)$ and conditional on the failure event $F$. The PDF of the
failure sample, denoted by $(\btheta', U')$, is given by
\begin{align}
p_{\btheta', U'}(\btheta, u) = P_F^{-1} \, q(\btheta) \, I(0\leq u\leq1) \, I(u < c \L(\btheta)) , 
\label{eq:joint_failure}
\end{align}

\noindent where 
\begin{align}
P_F = \int \int q(\btheta) \, I(0\leq u\leq1) \, I(u < c \L(\btheta)) \, du\,d\btheta , 
\end{align}

\noindent is the {\it failure probability} of the reliability problem. 

In the above formulation, the driving response variable can be defined as
\begin{align}
Y = c \, \L(\btheta) - U, \label{eq:driving_var}
\end{align}

\noindent so that the failure event corresponds to 
\begin{align}
F = \{Y > 0\}. 
\end{align}

\noindent Populations of failure samples conditional on the intermediate failure
events $F_i = \{Y > b_i\}$ for adaptively increasing $b_i$ $(i = 1, 2, \ldots)$
are then generated until they pass the target failure event $F = \{Y > 0\}$,
from which the samples conditional on $F$ are collected as the posterior
samples.

Note that in the original formulation the driving response variable was in fact
defined $Y = U - c \, \L(\btheta)$. The presentation in \eqref{eq:driving_var}
is adopted so that it is consistent with the conventional SuS literature, where
the intermediate threshold levels increase rather than decrease as the
simulation level ascends.

\section{Likelihood multiplier}

One issue of concern in the BUS formulation is the choice of the multiplier $c$
satisfying the inequality in \eqref{eq:reject_c}, which is not always trivial.
Some suggestions were given, by inspecting the mathematical structure of the
likelihood function \citep{Straub2014}; or by adaptively using empirical the
information from the generated samples \citep{Betz2014}. The latter is more
robust as it does not require preliminary analysis, but, as stated by the
authors, in order to guarantee that it satisfies the inequality, more
theoretical analysis is needed. In this section we rigorously investigate more
fundamentally the role of the multiplier and its effect on the results if it is
not properly chosen. The investigation leads to a reformulation of BUS, to be
proposed in the next section.

In the context of BUS, the multiplier needs to be chosen before starting a SuS
run as it affects the definition of the driving variable $Y$ in
\eqref{eq:driving_var}. Clearly, the multiplier affects the distribution of the
driving variable as well as the generated samples. Recall that only those
samples conditional on $Y = c \, \L(\btheta) - U > 0$ are collected as the
posterior samples. The larger the value of $c$ the more efficient the SuS run,
because this will increase $Y$ and the failure probability $P(Y > 0)$, thereby
reducing the number of simulation levels required to reach the target failure
event.

From the inequality in \eqref{eq:reject_c}, the choice of the multiplier is
governed by the region in the parameter space of $\btheta$ where the value of
$\L(\btheta)$ is large. The largest admissible value of $c$ is given by
\begin{align}
c_{\max} = \left[ \text{max}_{\btheta} \, \L(\btheta) \right]^{-1}. \label{eq:c_max}
\end{align}

\noindent This result is well-known in the rejection sampling literature.
Clearly, this value is not known before computation. While using a value smaller
than $c_{\max}$ will be less efficient but still give the correct distribution
in the samples, using a value larger than $c_{\max}$ will lead to bias in the
distribution of the samples. In some problems it is possible to investigate the
mathematical structure of $\L(\btheta)$ and derive inequalities to propose a
choice of $c$ that guarantees $c \L(\btheta) \leq 1$. In such cases, it is
computationally beneficial to use that value. However, in general it is
difficult by numerical means to have a choice of $c$ that guarantees the
inequality.

When an inadmissible (too large) value of the multiplier is used, the resulting
distribution of the failure samples will be truncated, leading to bias in the
posterior statistical estimates based on them. To see this, note that the
inequality \eqref{eq:reject_c} was used in establishing the third equality in
\eqref{eq:marginal_failure}. Suppose this inequality is violated, say, within
some region $B$:

\begin{align}
B = \{ \btheta \in \R^n : c\L(\btheta) > 1\}. 
\end{align}

\noindent Then for any $\btheta \in B$, $I(u < c \, \L(\btheta) ) = 1$ for $ u
\in (0,1)$ and so \eqref{eq:marginal_failure} implies
\begin{align}
p_{\bs \Theta}(\btheta) = P_F^{-1} \, q(\btheta) \, \int_0^1 \, I(u < c \L(\btheta)) du = P_F^{-1} 
q(\btheta).
\end{align}

\noindent For those $\btheta$ not in $B$, the inequality is satisfied and the
PDF value $p_{\bs \Theta}(\btheta)$ remains to be the correct posterior PDF
$p_{\D}(\btheta)$ as in \eqref{eq:marginal_failure}:
\begin{align}
p_{\bs \Theta}(\btheta) =  P_F^{-1} \, q(\btheta) \, c \L(\btheta) \propto p_{\D}(\btheta).
\end{align}

Thus, an inadmissible (too large) value of $c$ introduces bias in the problem by
truncating the posterior PDF to be the prior PDF in the region of $\btheta$
where the inequality is violated. Intuitively, in the context of rejection
principle, if the multiplier is not small enough, the samples drawn from the
prior PDF are accepted (incorrectly) {\it too often}, rendering their
distribution closer to the prior PDF than they should be.

The truncation effect is illustrated in Figure \ref{fig:truncation}, where the
shaded interval denotes the region $B$. The prior PDF $q(\btheta)$ is taken to
be constant and so $p_{\D}(\btheta) \propto c\L(\btheta)$. Instead of the target
posterior PDF, the resulting distribution of the sample takes the shape of the
center line. Within the region $B$ it is truncated to the shape of $q(\btheta)$.

\begin{figure}[!htp]
\centering
\includegraphics[draft=false,width =.6\linewidth]{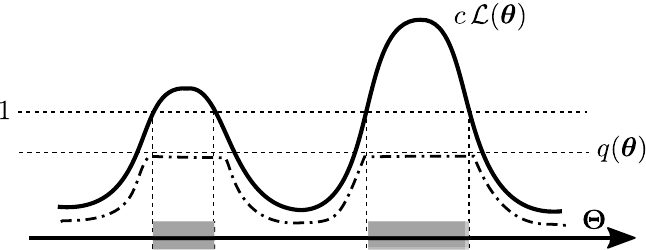}

\caption{Truncation of distribution in rejection algorithm. Center line - resulting
distribution (short of the constant $P_F^{-1}$); shaded interval - truncation
region $B$ where $c\L(\btheta) > 1$. }

\label{fig:truncation}

\end{figure}

As long as the multiplier satisfies the inequality in \eqref{eq:reject_c}, it is
completely arbitrary and it does not affect the distribution of the resulting
samples, which is equal to the correct posterior PDF. This observation is
trivial but has important implications. In the original BUS context, for
example, it implies that the samples generated in different simulation runs with
different admissible values of the multiplier can be simply averaged for
estimating posterior statistics, because they all have the same correct
posterior distribution. This fact shall also be used later when developing the
proposed algorithm in this work.

\section{Alternative BUS formulation}

Having clarified the role of the multiplier, we now present a modification of
the original BUS formulation that isolates the effect of the multiplier in a
fundamental manner. This leads to a formulation where SuS can be performed
without having to choose the multiplier before the simulation run; and where the
effect of the multiplier appears clearly in the accuracy of the posterior
distribution. The modification is based on the simple observation that the
failure event in \eqref{eq:equiv_rel} can be rewritten as
\begin{align}
F = \left\lbrace \ln \left[ \frac{\L(\bs \Theta)}{U} \right] > - \ln c \right\rbrace. 
\end{align}

\noindent This means that the driving variable in SuS can be defined as
\begin{align}
Y = \ln \left[ \frac{\L(\bs \Theta)}{U} \right], \label{eq:driving_variable}
\end{align}

\noindent and the target failure event can now be written as
\begin{align}
F = \{Y > b\}, 
\end{align}

\noindent where 
\begin{align}
b = - \ln c. \label{eq:driving_var_exp}
\end{align}

\noindent The base of the logarithm is arbitrary but we choose to use natural
logarithm here to facilitate the analysis.

Despite the apparently slight change in definition of the driving variable, the
setup above changes the philosophy behind the multiplier and the way SuS is
implemented to produce the posterior samples. The driving variable no longer
depends on the multiplier and so the choice of the latter is no longer needed to
start the SuS run. The multiplier only affects the target threshold level $b$
beyond which the samples can be collected as posterior samples. As remarked at
the end of the last section, as long as the multiplier is sufficiently small to
satisfy the inequality in \eqref{eq:reject_c}, the distribution of the samples
conditional on the failure event $F =\{U < c\L(\btheta)\}$ is invariably equal
to the posterior distribution. This implies that in the proposed formulation the
distribution of the samples conditional on $\{Y > b\}$ will settle (remain
unchanged) for sufficiently large $b$. In the original BUS formulation where the
driving variable is defined as $Y = c\L(\btheta) - U$ in
\eqref{eq:driving_var_exp} for a particular value of $c$ (assumed to be
admissible), only the samples conditional on the failure event $F = \{Y > 0\}$,
\ie for a threshold value of exactly zero, have the posterior distribution.

Substituting $b = -\ln c$ from \eqref{eq:driving_var_exp} into
\eqref{eq:reject_c} and rearranging, the inequality constraint in terms of $b$
is given by, for all $\btheta$,
\begin{align}
b > \ln \L(\btheta). 
\end{align}

\noindent From \eqref{eq:c_max}, the maximum admissible value of $c$ is
$c_{\max} = \left[ \text{max}_{\btheta} \, \L(\btheta) \right]^{-1}$.
Correspondingly the minimum value of $b$ beyond which the distribution of
samples will settle at the posterior PDF is
\begin{align}
b_{\min} = - \ln c_{\max} = \ln \left[
\text{max}_{\btheta} \, \L(\btheta) \right]. 
\end{align}

\noindent Similar to $c_{\max}$, the value of $b_{\min}$ is generally unknown
but this does not affect the SuS run. Under the proposed formulation, one can
simply perform SuS with increasing levels until one determines that the
threshold level of the highest level has passed $b_{\min}$. Despite not knowing
$b_{\min}$, this turns out to be a more well-defined task as it is shown later
that the CCDF of $Y$, \ie $P(Y > b)$ versus $b$, has characteristic behaviour
for $b > b_{\min}$ .

The logarithm in the above formulation is introduced for analytical and
computational reasons, so that the driving variable is a well-defined random
variable. In particular
\begin{align}
Y = \ln \left[ \frac{\L(\btheta)}{U} \right] = \ln \L(\btheta) + \ln(U^{-1}). 
\end{align}

\noindent For $U$ uniformly distributed on $[0,1]$, $\ln(U^{-1})$ is
exponentially distributed with mean 1. For a well-posed likelihood function
$\L(\btheta)$ one can expect that $\ln \L(\btheta)$ is a well-defined random
variable when $\btheta$ is distributed as $q(\cdot)$ , and so is the driving
variable $Y$. In particular, if the first two moments of $\ln \L(\btheta)$ are
bounded, then the same is also true for the first two moments of $Y$ because

\begin{align}
E[Y] &= E[\ln \L(\btheta) + \ln U^{-1}] \nonumber \\
&= E[\ln \L(\btheta) ] + 1, \\
E[Y^2] &=  E\{[\ln \L(\btheta) + \ln U^{-1}]^2\} \nonumber \\
&= E\{[\ln \L(\btheta)]^2\} + 2 E[\ln \L(\btheta)]E[\ln U^{-1}] + E\{[\ln U^{-1}]^2\} \nonumber \\
&= E\{[\ln \L(\btheta)]^2\} + 2 E[\ln \L(\btheta)] + 2, 
\end{align}

\noindent since $E[\ln U^{-1}] = 1$ and $E\{[\ln U^{-1}]^2\} = 2$ (properties of the
exponential variable $\ln U^{-1}$). 

The authors believe that, while respecting the originality of BUS, the proposed
formulation resolves the issue with the multiplier, as the requirement of
choosing it a priori in the original formulation has been eliminated. The
theoretical foundation of the proposed formulation is encapsulated in the
following theorem.

\begin{theorem} \label{theo:admissible}
Let $\btheta \in \R^n$ be a random vector distributed as $q(\btheta)$ and $U$ be a random variable
uniformly
distributed on $[0,1]$; with $\btheta$ and $U$ independent. Let $\L(\btheta)$ be a
non-negative scalar function of $\btheta$. Define $Y =
\ln[\L(\btheta) / U]$ and $b = -\ln c$, for $c \in \R$. Then, for any $b >
\ln[\text{max}_{\btheta} \L(\btheta)]$:
\begin{enumerate}    

\item The distribution of $\btheta$ conditional on $\{Y >
b\}$ is $p_{\D}(\btheta) = P_{\D}^{-1} \, q(\btheta) \, \L(\btheta)$ where
$P_{\D} = \int q(z)\, \L(z) \, dz $ is a normalizing constant; 

\item $P_{\D} = e^b \, P(Y > b)$.

\end{enumerate}
\end{theorem}

\begin{proof} 

In order to prove the first part of the above theorem, first note that events $\{Y >
b\}$ and $\{c \L(\btheta)> U\}$ are equivalent. Integrating out the uniform random variable from
the PDF of the failure sample given by equation \eqref{eq:joint_failure} gives:
\begin{align}
p_{\btheta'}(\btheta) & = \int_0^1 p_{\btheta', U'} (\btheta,u) \, du \nonumber \\
&= p_{F}^{-1} \, q(\btheta) \, \int_0^1 I(0\leq u \leq1) \, I(u < c \, \L(\btheta)) \, du 
\label{eq:marginal_posterior} \\
&= p_{F}^{-1}\, q(\btheta) \, c \L(\btheta) \nonumber \\
&\propto p_{\D}(\btheta).  \nonumber 
\end{align}


\noindent The result will be valid for any $c < [\text{max}_{\btheta} \L(\btheta)]^{-1}$, or
equivalently for any $b > \ln[\text{max}_{\btheta} \L(\btheta)]$.

\noindent For the second part of the theorem, since $Y = \ln[\L(\btheta)/U]$ and
$(\btheta,U)$ has a joint PDF $q(\btheta) I(0 < u < 1)$, $P(Y > b)$
is given by
\begin{align}
P(Y > b) &= \int \int q(\btheta) \, I(0<u<1) \, I \left( \ln\left[ \frac{\L(\btheta)}{u}\right]> b
\right) 
\, du \, d\btheta \nonumber \\
&= \int q(\btheta) \int_0^1 I(u < e^{-b} \L(\btheta)) \, du\,d\btheta \\
&= e^{-b} \int q(\btheta) \L(\btheta) \, d\btheta \nonumber, 
\end{align}

\noindent since $\int_0^1 I(u < e^{-b} \L(\btheta)) \, du = e^{-b}\L(\btheta)$
when $e^{-b}\L(\btheta) < 1$ for all $\btheta$ ($b$ is admissible). Observe,
from the definition of the posterior \eqref{eq:posterior}, that $P_{\D}$ is
simply the last integral in \eqref{eq:marginal_posterior}. Thus,

\begin{align}
P_{\D} = e^b P(Y > b) & & b > b_{\min}. \label{eq:evidence_prob}
\end{align}

\noindent That is, when $b > b_{\min}$, $P_{\D}$ can be obtained as a product of
$e^b$ and the failure probability $P(Y>b)$ it corresponds to.

\end{proof}

\section{Bayesian model class selection}

In addition to providing the posterior distribution and estimating the updated
expectation in \eqref{eq:expectation_model_data}, the posterior samples can be
used for estimating the normalizing constant $P_{\D}$ in \eqref{eq:normalising}.
This is the primary target of computation in Bayesian model class selection
problems, where competing models are rated. In this section we show how this can
be done using the conditional samples generated by SuS in the context of the
proposed formulation.

Let $b$ be an admissible threshold level, \ie $b > b_{\min}$ , so that the
samples conditional on $\{Y > b\}$ have the correct posterior distribution
$p_{\D}(\btheta)$. Consider the failure probability $P(Y > b)$, which can be
estimated using the samples in SuS. 

Note that equation \eqref{eq:evidence_prob} can be rewritten as
\begin{align}
 P(Y > b) = e^{-b} P_{\D}  & & b > b_{\min}. \label{eq:evidence_prob_reag}
\end{align}

\noindent Since $P_{\D}$ is constant for a given problem, this suggests that for
sufficiently large $b$, $P(Y > b)$ will decay exponentially with $b$.
Interpreting $P(Y > b)$ as the CCDF of $Y$, this exponential decay gives a
picture similar to a typical CCDF encountered in reliability analysis. This is
another (though secondary) merit of introducing the logarithm in the definition
of the driving variable $Y$ in \eqref{eq:driving_variable}.

\section{Characteristic trends} \label{sec:charact_trends}

As shown in the last section, when $b > b_{\min}$ the failure probability $P(Y >
b) $ is theoretically related to the evidence $P_{\D}$
through \eqref{eq:evidence_prob}. In the actual implementation, $b_{\min}$ is not
known and so it is necessary to determine whether $b>b_{\min}$ so that the
samples conditional on $\{Y > b\}$ can be confidently collected as posterior samples. We argue that
the 
variation of $P(Y>b)$ with $b$ takes on
different characteristics on two different regimes of $b$. This can be used
to tell whether the threshold value of a particular simulation
level has already passed $b_{\min}$ in a SuS run, thereby suggesting a stopping criterion.

First, note that $P(Y > b)$ is a non-increasing function of $b$. When $b$ is at
the left tail of the CCDF, $P(Y>b) \approx 1$ and it typically decreases with
$b$, equal to $P_{\D}$ at $b > b_{\min}$. When $b >b_{\min}$, we know from
\eqref{eq:evidence_prob_reag} that $P(Y > b) = P_{\D} e^{-b}$ and so it decays
exponentially with $b$. We can thus expect that, as $b$ increases from the left
tail and passes $b_{\min}$, the CCDF of $Y$ typically changes from a decreasing
function to a fast (exponentially) decaying function. Correspondingly, the
function $\ln P(Y > b)$ changes from a slowly decreasing function to a straight
line with a slope of -1.

On the other hand, consider the following function:

\begin{align}
V(b) = b+\ln P(Y > b). \label{eq:evidence_function}
\end{align}

\noindent This function can be used for computing the log-evidence $\ln P_{\D}$ as it can be readily
seen 
that

\begin{align}
V (b) =  \ln P_{\D} & & b > b_{\min}. \label{eq:evidence_feasible}
\end{align}

\noindent When $b$ is at the left tail of the CCDF, $\ln P(Y > b) \approx 0$ and
so $V(b) \approx b$ increases linearly with $b$. The above means that as $b$
increases from the left tail of the CCDF of $Y$ the function $V (b)$ increases
linearly, going through a transition until it settles (remains unchanged) at
$\ln P_{\D}$ after $b > b_{\min} $. The characteristic behavior of $\ln P(Y > b)$
and $V (b)$ are depicted in Figure \ref{fig:trends}.

\begin{figure}[!htp]
\centering
\includegraphics[draft=false, width=.6\linewidth]{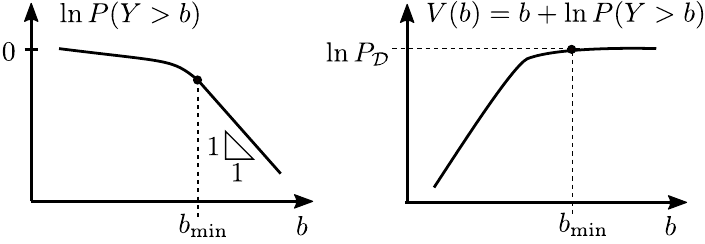}
\caption{Characteristic trends of $\ln P(Y > b)$ and $V(b)$.}
\label{fig:trends}
\end{figure}

Strictly speaking, the above arguments only apply to the theoretical quantities.
In a SuS run the quantities $\ln P(Y > b)$ and $V (b)$ as a function of $b$ can
only be estimated on a sample basis. The resulting estimated counterparts will
exhibit random deviation from the theoretical trends due to statistical
estimation error, whose extent depends on the number of samples used in the
simulation run (the larger the number of samples, the smaller the error).
Nevertheless, the above arguments and Figure \ref{fig:trends} provide the basis
for determining the simulation level to stop and to collect the posterior
samples, that is, once the trainsition in the slope of $\ln P(Y > b)$ and $V (b)$ is complete. On
this 
basis, we present an automatic stopping condition that is enforced 
once the algorithm detects that the transition has occurred.

\section{Automatic Stopping Strategy}
\label{s:stopping}

In the proposed context, the posterior samples can be obtained from the
conditional samples in a straightforward manner from a SuS run. No modification
of SuS is necessary. Below we outline how this can be done, focusing only on
issues directly related to the Bayesian updating problem.

The primary target of the Bayesian updating problem is to generate
posterior samples of $\bs \Theta$ distributed as the posterior PDF
$p_{\D}(\btheta) \propto q(\btheta)\L(\btheta)$, where $q(\btheta)$ is the prior
distribution assumed to be chosen from a standard class of distributions (e.g.,
Gaussian, exponential); and $\L(\btheta)$ is the likelihood function for a given
set of data. As reviewed in Section \ref{sec:sus}, a SuS run produces the
estimate of the CCDF of the driving variable $Y $, \ie $P(Y > b)$ versus $b$.
The posterior samples for Bayesian model updating can be obtained as the
conditional samples in a SuS run for the reliability problem with driving
variable $Y = \ln[\L(\btheta)/U]$, where $\btheta$ is distributed as
$q(\btheta)$ and $U$ is uniformly distributed on [0,1]; with $ \btheta$ and $U$ 
independent. The conditional samples are collected from the level whose
threshold level is determined to be greater than $b_{\min}$.

\subsection{Stopping criterion}

From the discussion in Section \ref{sec:charact_trends} and the definition of
SuS, it is clear that the intermediate  failure levels will continue to increase as the
algorithm progresses. For a given level $k$ where $b_k$ is an admissible value
for the failure event, the samples generated will eventually be distributed as desired.
The following theorem establishes theoretical guarantees that such failure level
can be achieved in a finite number of iterations, given some
regularity assumptions. Moreover, it provides a stopping criterion to terminate
the algorithm and prevent the generation of unnecessary SuS levels.

\begin{theorem}

Let the Bayesian inference problem be defined by an upper-bounded likelihood
function $\L(\bs\theta)$, a prior density
$q(\bs\theta)$ and associated posterior $p(\bs\theta | \D)$. The marginal
distribution of $\bs\theta$ conditional on the intermediate failure levels, denoted by $p(\bs\theta
| F_k)$,
converges to the posterior. Moreover, there exist constants $e^{-b_k}$ and a monotone decreasing
sequence $a_k$, such that

\begin{align}
\lim_{k \rightarrow \infty } a_k = 0.
\end{align}
where $a_k$ is the prior probability of the set  $B_k = \{ \bs\theta : e^{-b_k} \L(\bs\theta) > 1
\}$.
\end{theorem}

\begin{proof}

In Theorem \ref{theo:admissible}, it was proved that as long as the $j$-th failure level
satisfies $b_j > b_{\min}$, any sample generated will be distributed according
to the target posterior distribution. The level $b_j$ is said to be a terminal
level since any value of $b_{j+1}$ is, by definition, $b_{j+1} > b_{j}$. Hence,
the samples will be distributed as desired for any terminal level.

To prove the theorem, let us characterise a non-terminal level $k$ such that
$b_k < b_{\min}$. For the optimal threshold level $b_{\min}$, the inequality

\begin{align}
u < e^{-b_{\min}}  \L(\bs\theta) < 1, \label{eq:ineq_thresh}
\end{align}

\noindent is guaranteed for any value of a failure sample $(\bs\theta,u)$ being
distributed jointly as equation \eqref{eq:joint_failure}. In contrast, a 
non-terminal level satisfies $ e^{-b_{\min}} \L(\bs\theta) < e^{-b_k}
\L(\bs\theta)$ and it is not possible to determine an analogous right-hand side
of inequality \eqref{eq:ineq_thresh}. Let the inadmissible set be defined
as $B_k = \{ \bs\theta : e^{-b_k} \L(\bs\theta) > 1 \}$. It follows that the
marginal distribution of the target variable is given by

\begin{eqnarray}
p(\bs\theta | F_k) \propto  \left\lbrace\begin{array}{ll} 
q(\bs\theta) & \text{if} \, \bs\theta \, \in \, B_k \\ \\
e^{-b_k} \, q(\bs\theta) \, \L(\bs\theta) & \text{if} \, \bs\theta \, \in \, B_k^{\mathsf{c}} .
\end{array} \right. 
\end{eqnarray}

\noindent Note that for all samples in the inadmissible set $B_k$, the marginal
is proportional to the prior distribution, whilst for the samples in the
admissible set $B_k^{\mathsf{c}}$ the target density is proportional to the
posterior distribution. Marginalising in order to compute the normalising
constant results in

\begin{align}
P_{F_k} &= \int_{\bs\Theta} \, \left[ \, q(\bs\theta) \, I (\bs\theta \in {B_k}) + e^{-b_k} \,
q(\bs\theta) 
\, \L(\bs\theta) \, I ( \bs\theta \in B_k^{\mathsf{c}} )\right] \, d\bs\theta \nonumber \\
 &= \int_{B_k} q(\bs\theta) \, d\bs\theta + e^{-b_k} \,  \int_{B_k^{\mathsf{c}}} q(\bs\theta) \, 
\L(\bs\theta) \, d\bs\theta \nonumber \\
 &= P_{\bs\theta} ( B_k )  + e^{-b_k} \, P_{\D} \, P_{\bs\theta | \D}  (B_k^{\mathsf{c}}) , 
\label{eq:constants}
\end{align}

\noindent where $P_{\bs\theta} ( B_k ) $ denotes the probability of event $B_k$
under the prior distribution and $P_{\bs\theta | \D} (B_k^{\mathsf{c}})$ denotes
the probability of event $B_k^{\mathsf{c}}$ under the posterior distribution. Note that
equation \eqref{eq:constants} is consistent with the case where $b_k$ is a
terminal level. If that is the case, the pair $(\bs\theta, u)$ satisfies $u <
e^{-b_k} \L(\bs\theta)$ by the definition of the driving variable $Y$ and thus
$B_k = \varnothing$. Let us rewrite the inadmissible set as 
\begin{align}
B_k = \{ \bs\theta : \L(\bs\theta) > e^{b_k} \}.
\end{align}

\noindent Given an increasing sequence of failure levels, it can be seen that the
sequence of inadmissible sets is monotone decreasing, namely
\begin{align}
    B_{k} \supset B_{k+1} \supset \ldots \supset \varnothing.
\end{align}

\noindent This fact is depicted in Figure \ref{fig:admissible_set}. 

\begin{figure}[H]
\centering
\includegraphics[draft=false,width=.5\textwidth]{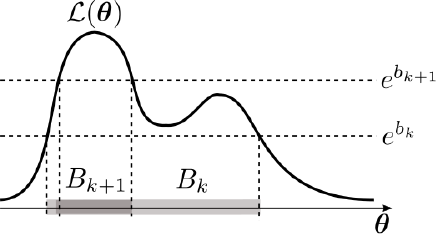}
\caption{Increasing failure levels and likelihood.}
\label{fig:admissible_set}
\end{figure}

\noindent Additionally, since the prior distribution is a probability measure,
it satisfies the monotonicity property, namely $P(B_{k+1}) \leq P(B_{k})$ for
all $k$. Let us define the sequence $a_k$ as the prior probability of the
inadmissible sets, \ie $a_k = P_{\bs\theta} ( B_k )$. As a consequence of the
monotonicity property, it follows that $a_k$ is a monotone decreasing
sequence of values converging to zero from above, denoted by

\begin{align}
a_k \searrow 0. \label{eq:stop01}
\end{align}

Moreover, since the sets $B_k$ are monotone decreasing, then the
sequence of complements is increasing, that is
\begin{align}
B_{k}^{\mathsf{c}} \subset B_{k+1}^{\mathsf{c}} \subset \ldots \subset \Theta.
\end{align}

\noindent Let $m_k$ denote the posterior probability of the set
$B_{k}^{\mathsf{c}}$. Analogous to $a_k$, the sequence $m_k$ is monotone
increasing converging to 1 from below. This is denoted by
\begin{align}
m_k \nearrow 1. \label{eq:stop02}
\end{align}

Expressions \eqref{eq:stop01} and \eqref{eq:stop02} allow to establish that
for a sufficiently large value of $k$
\begin{align}
p_{F_k} = e^{-b_k} \, P_{\D}, 
\end{align}

\noindent is satisfied and the result is established.

\end{proof}

The preceding theorem allows us to propose a stopping criterion for the BUS
algorithm with driving variable $ Y = \log[ \L (\bs\theta ) /u ] $ using SuS . The value of $a_k$
can be 
made 
arbitrarily small by means of the failure level $b_k$, which is learnt automatically during the
algorithm. 
The computation of $a_k$ is challenging, since it involves a multiple integral. Note that the prior 
probability can be written as

\begin{align}
a_k = P_{\bs\theta} ( B_k ) = P_{\bs\theta} (\L(\bs\theta) > e^{b_k}) \label{eq:inner}
\end{align}

\noindent which is in itself a reliability problem, where the likelihood $\L(\bs\theta)$
takes the role of a performance function and $e^{b_k}$ is a reliability
threshold. Since the prior distributions are chosen from a standard catalogue of
density functions and the probability is assumed to be small, such integral can also
be computed by means of SuS. In this setting, computing equation
\eqref{eq:inner} can be regarded as performing an {\it inner level} SuS. The sampling of the 
expanded variables $(\bs\theta, u)$ from the
failure levels in equation \eqref{eq:joint_failure}, is regarded as {\it outer level} SuS.

\subsection{Posterior statistical estimation}

The posterior samples $\{\btheta^{(m)}_k : k = 1,\ldots, N\}$ obtained from
simulation level $m$ for which $b_m > b_{\min}$ can be used for estimating
posterior statistics in Bayesian updating problem and the evidence for Bayesian
model class section. For the former, the posterior expectation in \eqref{eq:expectation_model_data}
is
estimated by simple averaging:

\begin{align}
E[r(\btheta) | \D, \M)] \approx \frac{1}{N} \sum_{k = 1}^N r(\btheta_k^{(m)}). 
\end{align}

\noindent On the other hand, based on \eqref{eq:evidence_prob}, the evidence can be estimated by

\begin{align}
P(\D | \M) = P_{\D} \approx \hat{P}_{\D} = e^{b_m} p_0^m. \label{eq:evidence_subset}
\end{align}

\noindent Taking logarithm, the log-evidence is estimated by

\begin{align}
\ln P(\D | \M ) = \ln P_{\D} \approx \ln \hat{P}_{\D} = b_m + m \ln p_0. 
\label{e:estlogev}
\end{align}

\subsection{Statistical error assessment}

Some comments are in order regarding the statistical error of the results, in
terms of the quality of the posterior samples and the statistical variability of
the log-evidence estimator. Provided that the threshold value of the simulation
level is greater than $b_{\min}$, its conditional samples are always distributed
as the target posterior PDF $p_{\D} (\btheta)$. \added{As MCMC samples they are
correlated, however}. When used for statistical estimation they will give less
information compared to if they were independent. Typically their correlation
tends to increase with the simulation level. In view of this, it is not
necessary to perform more simulation levels than necessary. The stopping 
criterion based on the inner-outer procedure discussed above guards against this scenario.

For the evidence estimate in \eqref{eq:evidence_subset}, it should be noted that
its statistical variability arises from $b_m$. By taking small random
perturbation of the estimation formula, it can be reasoned that
$\text{c.o.v.}(\ln P) \approx \text{std}(\ln P) \approx \text{std}(b_m)$, where
std is an abbreviation for standard deviation. An estimation formula for
the c.o.v. of $b_m$ based on samples in a single SuS run is not available,
however. Conventionally only the c.o.v. of the estimate $\hat{P}_b$ (say) for
$P(Y > b)$ for fixed $b$ is available, rather than the c.o.v. of the $b$
quantile value $b_m$ for fixed exceedance probability. It can be reasoned,
however, that the c.o.v. of $\hat{P}_{\D}$ (where $b_m$ is random) can be
approximated by the c.o.v. of $e^b\hat{P}_b$ for fixed $b$ (then taking $b =
b_m$ obtained in a simulation run). The latter is equal to the c.o.v. of
$\hat{P}_b$, for which standard estimation formula is available
\citep{Au2001,Au2014}.

\subsection{Comparison with original BUS formulation}

Table 1 provides a comparison between the original BUS and the proposed formulation.
Implementing SuS under the proposed framework has several advantages over the
original BUS, stemming mainly from the treatment of the multiplier in the
former. First of all, there is no need to determine the appropriate value of the
multiplier to start the simulation run. The definition of the driving variable
is more intrinsic as it only depends on the likelihood function and not on the
multiplier. In the BUS context, if the chosen value of the multiplier is not
small enough, it will lead to bias in the distribution of the samples,
unfortunately in the high likelihood region of the posterior distribution that
is most important. If it is chosen too small it will result in lower efficiency,
as it requires more simulation levels to reach the target event from which the
samples can be taken as posterior samples. In both cases if it is found after a
SuS run that the choice of the multiplier is not appropriate, one needs to
perform an additional run with a (hopefully) better choice of the multiplier.
These issues are all irrelevant in the proposed context because the problem
specification of the SuS run does not depend on the multiplier.

\begin{table}[H]
\begin{tabular}{lll}
\hline
 & BUS & Proposed \\
\hline
Driving variable &  $ Y = c \L(\btheta) - U $ & $Y = \ln [\L(\btheta)/U]$ \\
& for any $ c < [\max_{\btheta}\L(\btheta)]^{-1}$ &  \\
\\
Target failure event & $F = \{Y > 0\}$ & $F = \{Y > b\}$ \\
& & for any $b>\ln [\max_{\btheta} \L(\btheta)]$ \\
\\
Evidence calculation & $P_{\D} = c P(Y > 0 )$ & $P_{\D} = e^b P(Y > b)$ \\
& & for any $ b > \ln [\max_{\btheta} \L(\btheta)]$\\
\\
Stopping criterion & When threshold value of & \added{After inner-outer SuS procedure} \\
&  simulation level is equal & \added{automatically determines that the} \\
& to zero.& \added{threshold $b_{min}$ has been crossed by} \\
& & \added{driving the sequence $a_k \searrow 0$.}
\end{tabular}
\centering
\caption{Comparison of original BUS and proposed reformulation. Note that the original definition of
the 
driving variable in BUS is $Y = U -  c\L(\btheta)$. For consistency with SuS literature, it has
been 
reexpresed as shown here.}
\end{table}

On the other hand, in the BUS context the posterior samples must be obtained as
those conditional on the target failure event $\{Y > 0\}$ where $Y =
c\L(\btheta) - U $. For example, samples conditional on ${Y > 0.1}$ cannot be
directly used. Since the threshold values ${b_1 ,b_2 ,\ldots}$ generated
adaptively in different simulation levels of SuS are random, they generally do
not coincide with 0 , \ie the target threshold value of interest. In this case,
not all samples can be used directly as conditional samples. In the original BUS
algorithm if the threshold level of the next level
determined adaptively from the samples of the current level is greater than
zero, it is set equal to zero so that the next (and final) level is exactly
conditional on $\{Y > 0\}$. In the proposed context, the posterior samples can
be directly collected from the samples generated in SuS. This is because any
sample conditional on $\{Y > b \}$ with $b>b_{\min}$ $(Y  = \ln
[\L(\btheta)/U])$ can be taken as a posterior sample. The value of
$b_{\min}$ is unknown but whether $b > b_{\min}$ can be determined from the
inner-outer procedure discussed in Section \ref{s:stopping}.
\section{Illustrative examples}

We now present two examples that illustrate the applicability of the proposed
methodology. The first one is the locally identifiable case of a two-degree-of-
freedom shear building model originally presented in \cite{Beck2002}. The
second example is the unidentifiable case of the same model.

\subsection{Example 1. Two-DOF shear frame: locally identifiable case}

Consider a two-storied building structure represented by a two-degree-of-freedom
shear building model. The objective is to identify the interstory
stiffnesses which allow the structural response to be subsequently updated. The
first and second story masses are given by $16.5\times10^3$ kg and $16.1 \times
10^3$ kg respectively. Let $\btheta = [\theta_1 ,\theta_2 ]$ be the stiffness
parameters to be identified. The interstory stiffnesses are thus parameterized
as $k_1 = \theta_1 \overline{k}_1$ and $k_2 = \theta_2 \overline{k}_2$, where
the nominal values for the stiffnesses are given by $k_1 = k_2 =29.7 \times
10^6$ N/m. The joint prior distribution $q(\cdot)$ for $\theta_1$ and $\theta_2$
is assumed to be the product of two Lognormal distributions with most probable
values 1.3 and 0.8 respectively and unit standard deviations. For further
details on the assumptions behind the parameterization and the choice of nominal
values, refer to \cite{Beck2002}. Let $ \D = \{\tilde{f}_1, \tilde{f}_2\}$ be
the modal data used for the model updating, where 3.13 Hz and 9.83 Hz are the
identified natural frequencies. The posterior PDF is formulated following
\cite{Vanik2000} as

\begin{align}
p_{\D} \propto \exp [- J(\btheta) / 2\epsilon^2] \, q(\btheta),
\end{align}

\noindent where $\epsilon$ is the standard deviation of the prediction error and
$J(\btheta)$ is a modal measure-of-fit function given by

\begin{align}
J(\btheta) = \sum_{j = 1}^2 \lambda_j^2 \, [f_j^2(\btheta)/\tilde{f}_j^2 -1 ].
\end{align}

\noindent Here, $\lambda_1$ and $\lambda_2$ are weights and $f_1(\btheta)$ and
$f_2(\btheta)$ are the modal frequencies predicted by the corresponding finite
element model.

For the implementation of SuS, a conventional choice of algorithm parameters in
the reliability literature is adopted in this study. The level probability is
chosen to be $p_0 = 0.1$ and the number of samples per level $N$ is fixed at 10,000.
In the standard Gaussian space, the one-dimensional proposal PDF is chosen to be
uniform with a maximum step width of 1. A relatively large number of samples per
level is been chosen in this study to illustrate the theoretical aspects of the
proposed method. Strategies for efficiency improvement such as adaptive proposal
PDF or likelihood function can be explored but are not further investigated
here. 

Figure \ref{fig:lognorm-space} shows the Markov chain samples for $\btheta =
[\theta_1, \theta_2]$ at six consecutive simulation levels. The results are
shown in the Lognormal space after the application of the relevant
transformation. Level 0 corresponds to the unconditional case (\ie Direct Monte
Carlo), that is, the joint prior PDF. As the simulation level ascends, the
distribution of the samples evolves from the prior distribution to the target
posterior distribution, which is bimodal in the present example.

\begin{figure}[!htp]
\centering
\includegraphics[width=.8\linewidth]{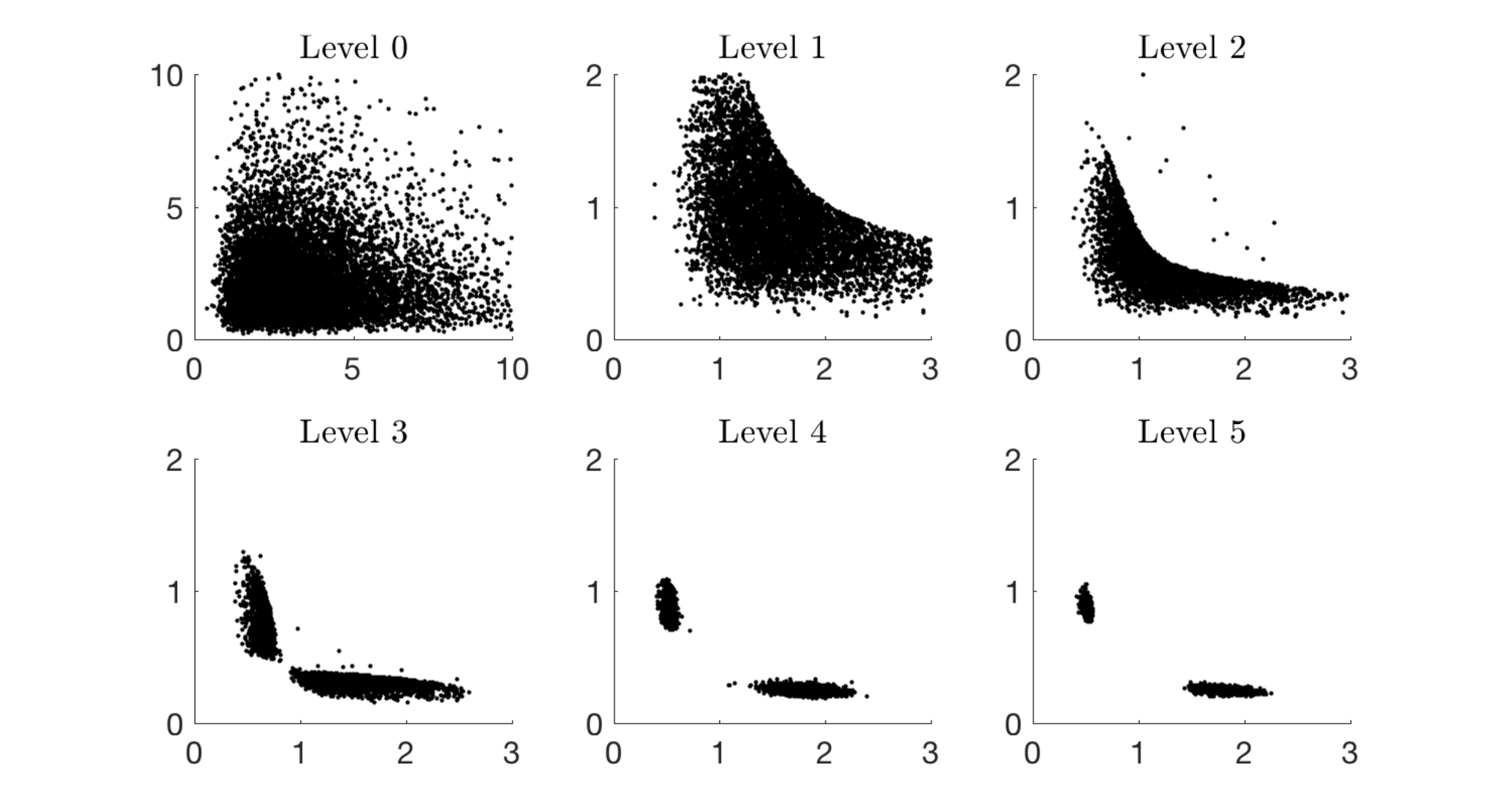}
\caption{Markov chain samples in the Lognormal space for the stiffness
parameters $\theta = [\theta_1 ,\theta_2]$ from Level 0 (prior distribution) to
Level 5.}
\label{fig:lognorm-space}
\end{figure}

Figure \ref{fig:post-samples-02} shows the marginal histograms for $\theta_1$ and
$\theta_2$ corresponding to those samples in Figure \ref{fig:lognorm-space}. For
comparison, the solid lines show the target marginal posterior distributions
obtained by numerically integrating the expression for the posterior PDF, which is still feasible
for this 
two-dimensional example. It is
apparent that the distribution of the samples has settled either in Level 4 or Level 5. In
reality, the exact target PDF is not available and so alternative means must be employed to
determine 
whether the distribution of the samples has settled at the target one. Within the context of the
current
methodology, this is done through the proposed automatic stopping strategy and confirmed by the
plots of the 
log-failure probability and log-evidence versus the threshold level. 

\begin{figure}[htp]
\centering
\includegraphics[width=.8\linewidth]{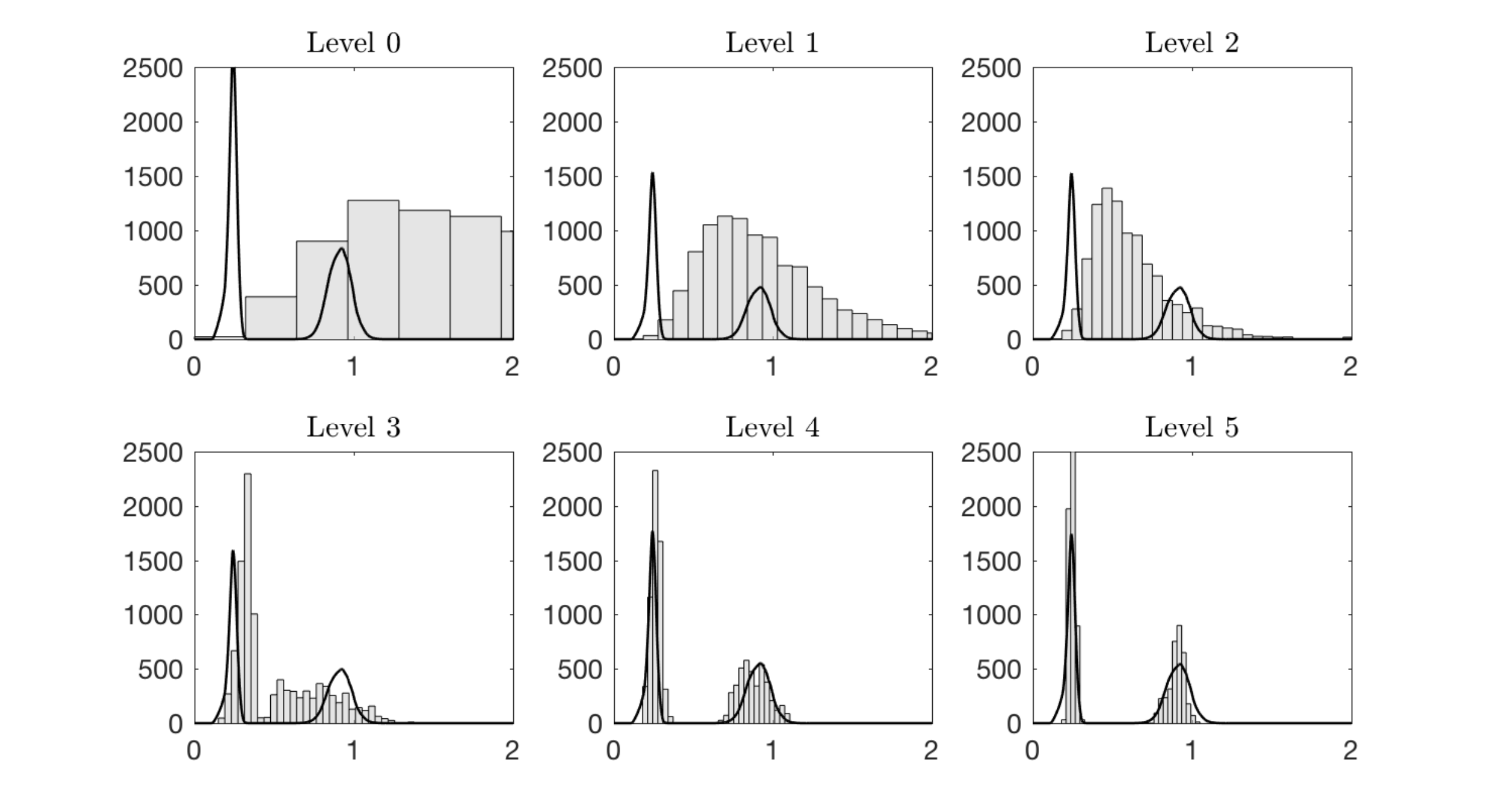}
\caption{Posterior marginal PDF for $\theta_2$ at different simulation levels. The target marginal 
posteriors were obtained numerically and are shown for comparison.}
\label{fig:post-samples-02}
\end{figure}

Figure \ref{subfig:ct-01} plots the estimate of the log-CCDF of $Y$, \ie $\ln
P(Y > b)$ versus $b$. The general shape of the resulting simulated curve
coincides with the characteristic trend predicted by the theory (see Figure \ref{fig:trends}), that
is, there
is a transition from a slowly decreasing function to a line with slope equal to
-1. When zooming into the region where $b > 0$, the figure shows the boundaries
of each level computed via SuS. Additionally, the log-evidence was computed following
\eqref{eq:evidence_function} and is shown in Figure \ref{subfig:ct-01}. As with
the log-CCDF, the theoretical prediction of the characteristic trend is also
verified for this case, whereby the curve flattens when the transition is complete. Table 
\ref{t:evol} shows the evolution of the threshold (columns 2 and 3). The transition is complete
after Level 4, {\color{black}where the probability of inadmissibility $a_k$ 
converges to zero (as defined in Section \ref{s:stopping}). For a tolerance of $a_k = 10^{-8}$, the
fourth column in Table \ref{t:evol} shows that the posterior samples should be collected from Level
5. This corresponds with the clearly bimodal distributions in figures \ref{fig:lognorm-space} and
\ref{fig:post-samples-02}. It is guaranteed that the samples in the subsequent Sus levels would all
be distributed according to the target posterior PDF. However, for statistical estimation their
quality deteriorates as the simulation level ascends because their correlation tends 
to increase. Thus, the algorithm stops in Level 5.}

\begin{figure}[!htp]
\centering
\includegraphics[width=.7\linewidth]{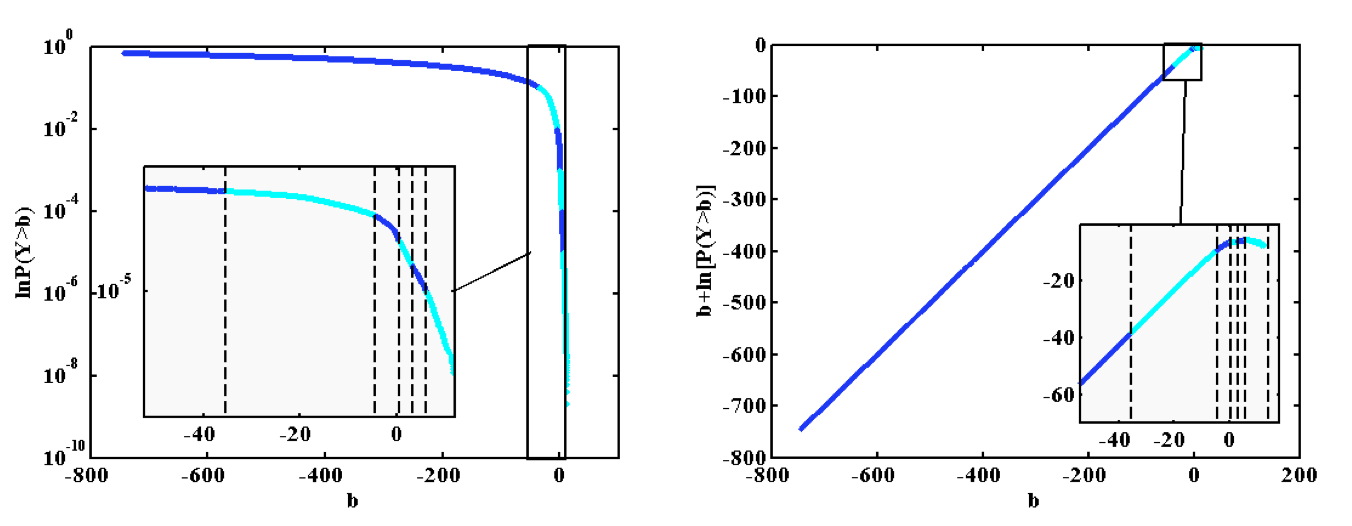} 

\caption{Log-CCDF computed through SuS (left plot) for the identifiable
case. The curve slowly transitions into a straight line with negative unit
slope. Correspondingly, the log-evidence (right plot) flattens as the
threshold exceeds $b_{\min}$. The dotted lines show the thresholds for different
simulation levels.}

\label{subfig:ct-01} 

\end{figure}

\begin{table}[H] 
\begin{tabular}{cccc}
\hline
Level & $b_k$ & $c_k$ &  $a_k$ \\
\hline
  0  & & & \\
  1 & -4.291e+02 & 2.325e+186 & 5.3300e-01 \\
  2 & -6.237e+01 &  1.221e+27 & 1.3800e-01 \\
  3 & -9.331e+00 &  1.128e+04 & 2.8700e-02 \\
  4 & 2.203e+00  & 1.105e-01 & 4.0400e-03 \\
  5 & 5.780e+00  & 3.088e-03 & 0.0000e+00
\end{tabular}
\centering
\caption{Evolution of the threshold and the probability of inadmissibility.}
\label{t:evol}
\end{table}

\subsection{Example 2. Two-DOF shear frame: unidentifiable case}

The exercise was repeated for the case where the story masses are also unknown
and need to be updated. The problem is characterized as unidentifiable, since
there are an infinite number of combinations of parameter values that can
explain the measured modal frequencies. In addition to the stiffnesses, the
masses are parameterized as $m_1 = \theta_3 \overline{m}_1$ and $m_2 = \theta_4
\overline{m}_2$, where the nominal values for the are given by $m_1 = 16.5
\times 10^3$ kg and $m_2 = 16.1 \times 10^3$ kg. Thus, for this case, $\btheta =
[\theta_1, \theta_2, \theta_3, \theta_4]$ where the marginal prior distributions
for $\theta_1$ and $\theta_2$ are the same Lognormals as in the previous
example. The prior marginal distributions for $\theta_3$ and $\theta_4$ are both
assumed to be Lognormals with most probable values equal to 0.95 and standard
deviation of 0.1. The joint prior PDF is therefore taken as the product of the
four Lognormals. Figure \ref{fig:post-samples-03} shows the Markov chain samples
for $\btheta$ ($\theta_1$ versus $\theta_2$ for visualization purposes) at
simulation levels 0 through 5. Again, the updated distribution results in a
bimodal posterior PDF.

\begin{figure}[!htp]
\centering
\includegraphics[width=.72\linewidth]{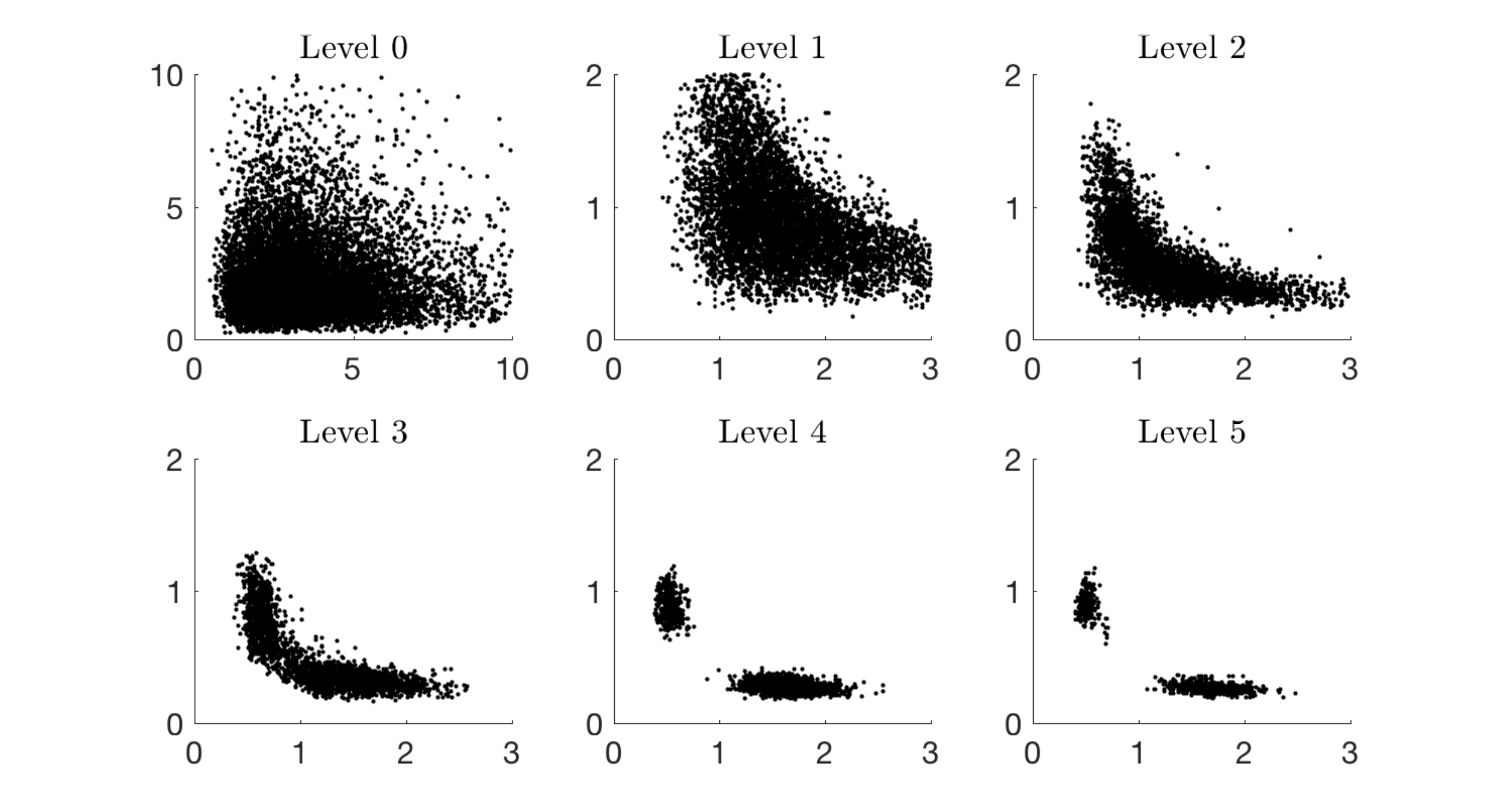}
\caption{Markov chain samples in the Lognormal space for the stiffness parameters $\theta_1$ and 
$\theta_2$ of the unidentifiable case at simulation levels 0 (prior distribution) to level 5.}
\label{fig:post-samples-03}
\end{figure}

Analogously, Figure \ref{fig:post-samples-04} shows the samples for $\theta_3$
and $\theta_4$ in the Lognormal space. There is no noticeable pattern in the
distribution of the masses, consistent with the findings in \cite{Beck2002}. The characteristics of
this 
example are very similar to the ones displayed by the locally identifiable case. The automatic
stopping 
condition is also reached when $a_k \leq 10^{-8}$, for which the posterior samples are also
collected in 
Level 5. We omit the characteristic trend plots and corresponding table for brevity.

\begin{figure}[!htp]
\centering
\includegraphics[width=.72\linewidth]{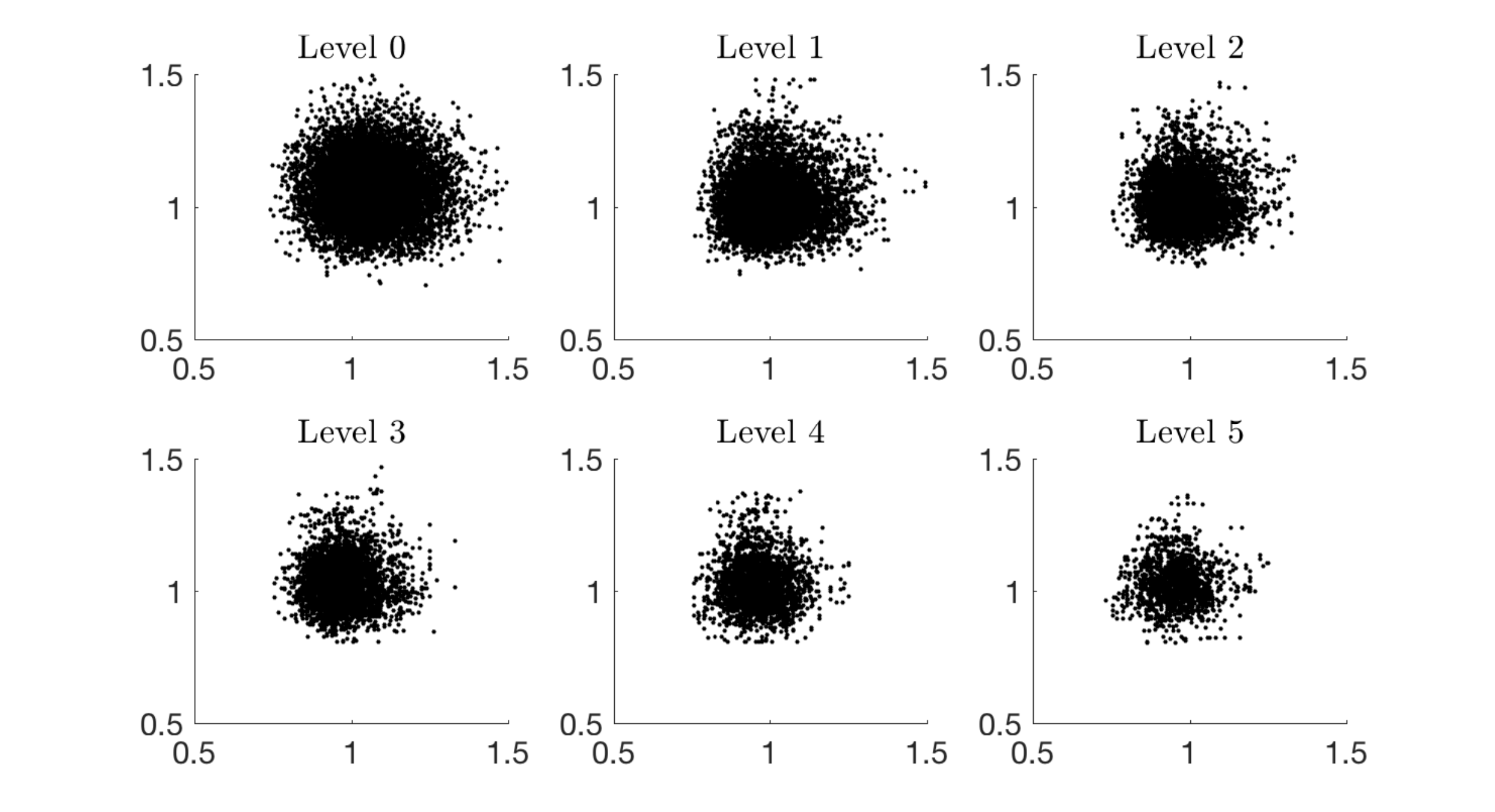}
\caption{Markov chain samples in the Lognormal space for the mass parameters $\theta_3$ and
$\theta_4$ of 
the unidentifiable example at simulation levels 0 to level 5.}
\label{fig:post-samples-04}
\end{figure}

\subsection{Example 3. Model Class Selection}

Following the two preceding examples, we can {\color{black}estimate the log-evidence
corresponding to each model according to equation \eqref{e:estlogev}.
Figure \ref{fig:model_class} shows the ratio of the evidence for the
identifiable case to the evidence of the locally unidentifiable case.}
Discounting the random deviation due to simulation error, the ratio of evidence
seems to converge to 1, which suggests that, given the available data, there is
no reason to prefer the unidentifiable model over the more
parsimonious one.

\begin{figure}[!htp]
\centering
\includegraphics[width=.35\linewidth]{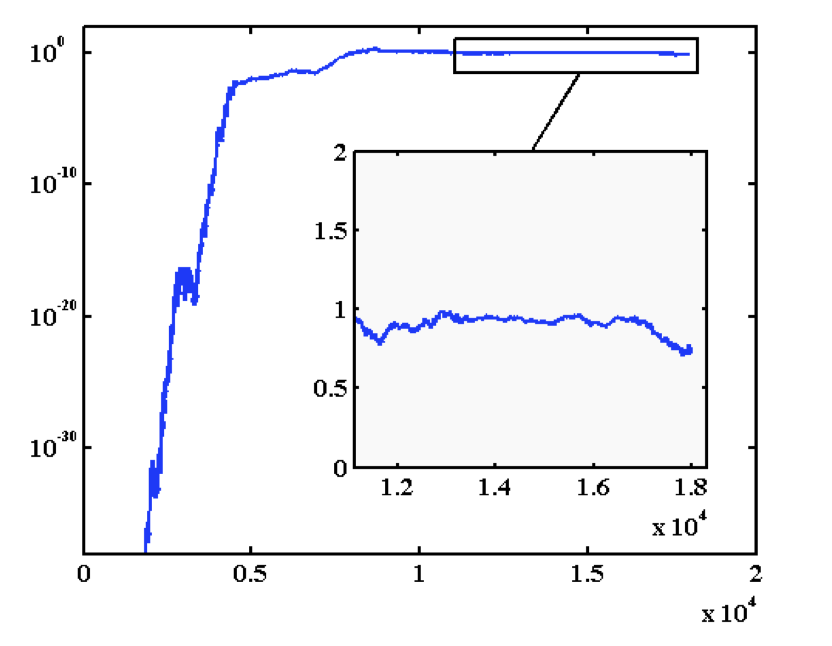}
\caption{{\color{black}Ratio of evidence of the identifiable model to the evidence of the locally 
unidentifiable model. Since this ratio converges to 1, there is no preference of either model over
each 
other, given the available data.}}
\label{fig:model_class}
\end{figure}

\section{Conclusions}

We have presented a fundamental analysis of BUS, a recently proposed framework
that establishes an analogy between the Bayesian updating problem and the
engineering reliability problem. This work was motivated by the question of choosing the
correct likelihood multiplier and it has led to an improved formulation which
resolves this question. By redefining the target failure event, we have
expressed the driving variable in the equivalent reliability problem using the
likelihood function alone, without the multiplier. This redefinition provides
the key advantage over the original BUS, since our implementation no longer
requires a predetermined value for the multiplier in order to start the SuS
runs. This immediately eliminates the need to perform additional runs in case an
inadmissible or inefficient value for the multiplier is chosen. Moreover, it was shown that the
samples generated at different levels of SuS can be used directly as posterior
samples as long as their threshold is greater than the minimum admissible
value and the probability of inadmissibility is zero. We have proposed an inner-outer SuS procedure
that 
provides an automatic stopping condition for the algorithm. The theoretical predictions of our
study have been verified by applying our proposed strategy to illustrative
examples.

\section{Acknowledgements}

This work was performed while the third author was on short-term fellowship
funded by the Japan Society for the Promotion of Science (JSPS) and hosted at
the Tokyo City University, whose supports are gratefully acknowledged. The
second author gratefully acknowledges the Consejo Nacional de Ciencia y
Tecnolog\'{i}a (CONACYT) for the award of a scholarship from the Mexican
government.

\section*{{\color{black}References}}

\end{document}